\documentclass [11pt]{article}

\usepackage{fullpage,wrapfig}
\usepackage{amsmath, amsthm, amssymb, algorithm,thmtools,algpseudocode,enumerate,caption,framed}
\usepackage{paralist, sidecap}

\usepackage[usenames,dvipsnames,svgnames,table]{xcolor}
\definecolor{darkgreen}{rgb}{0.0,0,0.9}

\usepackage[colorlinks=true,pdfpagemode=UseNone,citecolor=Blue,linkcolor=Red,urlcolor=BrickRed,
pagebackref]{hyperref}

\usepackage{tikz}

\makeatletter
\newcommand{\setword}[2]{%
  \phantomsection
  #1\def\@currentlabel{\unexpanded{#1}}\label{#2}%
}
\makeatother

\numberwithin{equation}{section}

\newtheorem{theorem}{Theorem}[section]

\newtheorem{lemma}{Lemma}[section]
\newtheorem{observation}{Observation}[section]

\title{Sliding $k$-Transmitters: Hardness and Approximation}

\author{
Therese Biedl, Saeed Mehrabi, and Ziting Yu\\
\small{Cheriton School of Computer Science}\\
\small{University of Waterloo, Waterloo, Canada.}\\
\small{{\texttt{\{biedl, smehrabi, z44yu\}@uwaterloo.ca}}}}

\date{}

\begin{document}

\maketitle

\begin{abstract}
A \emph{sliding $k$-transmitter} in an orthogonal polygon $P$ is a mobile guard that travels back and forth along an orthogonal line segment $s$ inside $P$. It can see a point $p\in P$ if the perpendicular from $p$ onto $s$ intersects the boundary of $P$ at most $k$ times. We show that guarding an orthogonal polygon $P$ with the minimum number of $k$-transmitters is \textsc{NP}-hard, for any fixed $k>0$, even if $P$ is simple and monotone. Moreover, we give an $O(1)$-approximation algorithm for this problem.
\end{abstract}

\section{Introduction}
\label{sec:introduction}
Art gallery problems are one of the standard problems in computational
geometry.  In the original setting, we are given a polygon (modelling
the art gallery) and we want to know a set of points (modelling guards
or cameras) that can see any point in the polygon, where ``see'' in the
original setting means that 
the line segment from the guard to the point is inside the polygon.
There have been numerous result, concerning bounds on the number of
guards needed, NP-hardness and approximation algorithms.  See 
e.g.~\cite{orourke1987,krohn2013} and the references therein.

Recently, motivated by covering a region with wireless transmitters,
Aichholzer et al.~\cite{aich09} introduced
variants where guards can see through
a limited number of walls. Hence a {\em $k$-transmitter}
is a point $p$ in a polygon $P$ that is considered to see all points $q$ 
in $P$ for which the line segment $\overline{pq}$ intersects the boundary
of $P$ at most $k$ times. Only cases of even $k$ are
interesting.  

We combine in this paper the concept of a $k$-transmitter with the
concept of a mobile guard.  A {\em mobile guard} is a guard that
is not stationary, but walks along a line segment $s$ inside the
polygon, and can see all points that are visible from some point of $s$.
For orthogonal polygons, a common restriction has been to demand that
line segment $s$ is horizontal or vertical, and that it guards only
those points $p$ that it can see in an orthogonal fashion, i.e., the
perpendicular from $p$ onto $s$ is inside $P$. This is
called a {\em sliding camera}.   We combine the concept of sliding
cameras with $k$-transmitters,
and hence define a {\em sliding $k$-transmitter} as follows:
It is a horizontal or vertical line segment $s$ inside an orthogonal polygon $P$
and it can see all points $p$ such that the perpendicular from $s$ onto
$p$ intersects the boundary of $P$ at most $k$ times.
We allow sliding $k$-transmitters to include edges of
the polygon.%
\footnote{With some minor modifications,  the results in this paper
also hold if guards must 
be strictly inside $P$ except at their end.} The objective is to guard $P$ with the minimum number of sliding $k$-transmitters.

\paragraph{Related Work.} Sliding cameras 
were introduced by Katz and Morgenstern~\cite{katz2011}.
Finding the minimum set of sliding cameras is
NP-hard in polygons with holes \cite{durocherM2013}, 
even if only horizontal sliding cameras
are allowed \cite{BiedlCLMMV}.  The optimum set of sliding cameras can be
found in polynomial time for monotone polygons \cite{deBergDM2014}.
The complexity for simple polygons is open.

Finding the minimum set of $k$-transmitters is NP-hard in simple polygons
\cite{cannon15}, regardless whether the transmitters are points or 
polygon-edges.  Numerous bounds are known on the number of 
$k$-transmitters that are necessary and sufficient, depending on the
type of transmitter (point or edge) and the type of polygon
\cite{aich09,aich15,Ballinger13,cannon15}.
Regarding sliding $k$-transmitters, an approximation algorithm for monotone polygons is claimed in~\cite{Mehrabi15}, but
the algorithm needs a minor modification to deal with an example (private communication); it is not clear
whether this modification suffices. Other optimization criteria for
sliding $k$-transmitters have also been considered~\cite{mahdaviSG2014}.

\paragraph{Our Results.} In this paper, we study the complexity of
finding the minimum set of sliding $k$-transmitters to guard an
orthogonal polygon.  Unsurprisingly, we can show that this is NP-hard,
but we prove NP-hardness even in a very restricted version:  The 
polygon is orthogonal and $y$-monotone, and
 there is an optimal solution with only horizontal sliding $k$-transmitters. 
We are not aware of {\em any}
other variant of the art gallery problem that is NP-hard on orthogonal monotone
polygons (the traditional art gallery problem is NP-hard for monotone
polygons \cite{krohn2013}, but slanted edges are crucial for the reduction
to work).  

As a second result, we show that the $O(1)$-approximation
algorithm that we recently developed for sliding cameras \cite{BiedlCLMMV}
works similarly for sliding $k$-transmitters.
Hence we have an $O(1)$-approximation for finding the minimum set
of sliding $k$-transmitters, in any (not necessarily simple) orthogonal
polygon.   The algorithm works also (and becomes even easier) if only
horizontal sliding $k$-transmitters are allowed.

\section{Hardness}
\label{sec:hardness}
In this section, we show that guarding with sliding $k$-transmitters
is NP-hard, even if the polygon is orthogonal and monotone (hence simple).
We first prove this for $k=2$ and then extend to larger $k$.

\subsection{Sliding $2$-Transmitters}
\label{sec:npHardForTwo}
We use a reduction from Minimum Vertex Cover in a graph $G$, which is known to be \textsc{NP}-hard even if $G$ is required to be planar and 2-connected (see e.g.~\cite{BKK97}).
So the objective is to compute a minimum set $C$ of vertices such that every edge has at least one endpoint in $C$. 

Given a planar 2-connected graph $G$ with $n$ vertices and $m$ edges, we first compute a \emph{bar visibility representation} of $G$ in which each vertex is assigned a horizontal line segment (called \emph{bar}) and for each edge there is a vertical \emph{strip} with positive width that connects the bars of endpoints and does not intersect other vertices.
It has been shown multiple times (see e.g.~\cite{TamassiaT1986}) that this exists and can be computed in linear time.  We may move vertex-bars up and down slightly as needed so that all vertex-bars have distinct $y$-coordinates.  Also, since edge-strips have positive width, we can make them thin enough such that no two of them have overlapping $x$-range.   Since the graph is 2-connected,
the construction in \cite{TamassiaT1986} guarantees that all vertices except
the bottommost one have a neighbour below, and all vertices except the
topmost one have a neighbour above. 

\paragraph{Gadgets.} 
We start by thickening each vertex-bar into a box, and place three copies of this box above each other with the same $x$-range.  
These three boxes are connected to each other by {\em channels}, which are thin vertical corridors (thin enough so that their $x$-range is strictly within that of the vertex-box, and does not intersect an edge-strip).
We place these two channels at opposite ends of the vertex-boxes, resulting in a $Z$-shape or an $S$-shape (the choice between the two is arbitrary for now, but will be determined later). We call the result a \emph{vertex-gadget}; see Fig.~\ref{fig:firstGadgets}.  
By making the height of boxes small enough, we may assume that no
two vertex-gadgets have overlapping $y$-range.  

For each edge $e$,  the \emph{edge-gadget} of $e$ is a small
axis-aligned box placed strictly within the strip representing $e$
in such a way that its $y$-range intersects no $y$-range of another 
(vertex- or edge-) gadget.
See Fig.~\ref{fig:firstGadgets}.  
Notice that from any edge-gadget there are vertical lines-of-sight
to the vertex-gadgets of the endpoints of the edge.

\paragraph{The Reduction.}
Let $P'$ be the polygon
obtained by replacing all vertex-bars and edge-strips with these gadgets.
$P'$ is $y$-monotone (i.e., any horizontal line intersects it in one
interval), but not connected (for now we allow the polygon to be disconnected, but we will discuss the modifications to make it connected later).

\begin{figure}[t]
\centering
\includegraphics[width=0.75\textwidth]{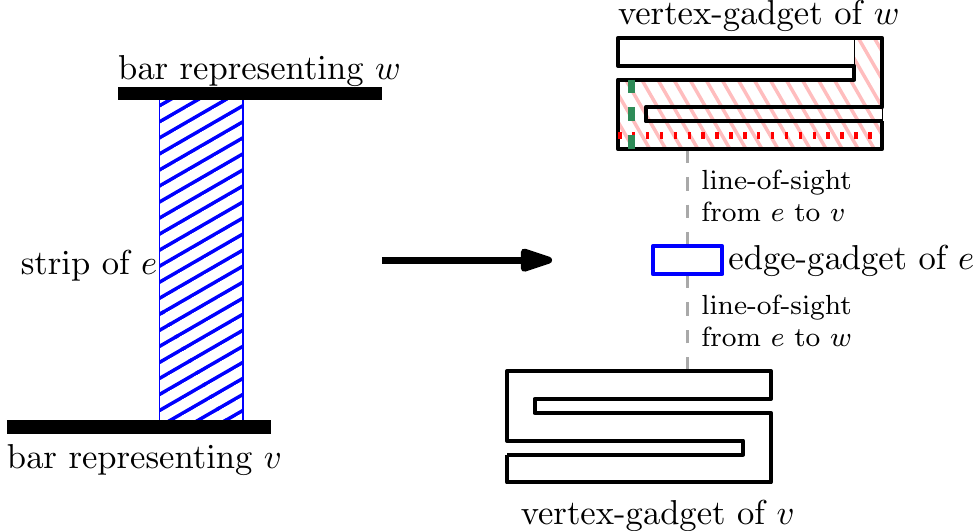}
\caption{Vertex- and edge-gadgets.  The pink (falling pattern)
region is guarded by the red (dotted) horizontal $2$-transmitter.
Note that it includes everything that the
green (dashed) vertical $2$-transmitter can see.
}
\label{fig:firstGadgets}%
\end{figure}

Since no $y$-ranges overlap, one can easily verify that vertical 2-transmitters
are never required.

\begin{observation}
\label{obs:no_vertical}
Any vertical sliding 2-transmitter in $P'$ can be replaced by 
a horizontal sliding 2-transmitter that guards at least as much.
\end{observation}

See also Fig.~\ref{fig:firstGadgets}.
We call the three boxes of a vertex-gadget
the {\em top, middle} and {\em bottom} box,
and also use {\em outer boxes} 
to mean the top and bottom box.

\begin{lemma}
\label{lem:standardForm}
\label{lem:transmitters_noEdge}
For any set $S$ of horizontal sliding 2-transmitters that guard $P'$ entirely, 
there exists a set 
$S'$ of horizontal sliding 2-transmitters that guard $P'$ entirely such that 
$\lvert S'\rvert\leq\lvert S\rvert$ and no sliding 2-transmitter of $S'$ is 
located in an edge-gadget.
\end{lemma}
\begin{proof}
Let $s\in S$ be a sliding 2-transmitter that lies in an edge-gadget $B$ corresponding 
to edge $e=(v,w)$. 
After possible renaming, assume that (the vertex-gadget corresponding to)
$v$ is below $e$ and $w$ is above $e$.

Assume first that one of $v,w$ (say $v$) has a horizontal sliding 
2-transmitter $s'$ in the outer box facing $e$.  After possibly extending $s'$
we may assume that it spans the entire outer box of $v$.
Since the $x$-range of
$B$ is within the $x$-range of $v$, $s'$ sees everything that $s$ saw and
that was below $s$.  So we can replace $s$
by a sliding 2-transmitter in the outer box of $w$ facing $e$, and
this  can only increase the guarded region.

So now assume that neither $v$ nor $w$ has a horizontal sliding 
2-transmitter
in the outer box facing $e$.  Consider a point $p$ in the top box of $v$
that is just outside the $x$-range of $B$, but still within the
$x$-range of $w$.  The only horizontal sliding 2-transmitters that could
guard $p$ are in the bottom box of $w$ or in the middle box of $v$.  By 
assumption we therefore have a sliding 2-transmitter
in the middle box of $v$.  Likewise $w$ must have a sliding 2-transmitter in the
middle box of $w$.  We can thus move the sliding 2-transmitter in $B$ to the 
bottom box of $w$ without decreasing the guarded region.
\end{proof}

\begin{lemma}
\label{lem:vertex_transmitters}
Let $S$ be a set of horizontal sliding 2-transmitters that guard $P'$ entirely
and that do not lie in edge-gadgets.
Then for any vertex $v$, there must be at least one sliding 2-transmitter
intersecting the vertex-gadget of $v$.  If there is exactly
one such sliding 2-transmitter, then it must be in the middle box of $v$.
\end{lemma}
\begin{proof}
Pick a point $p$ in the middle box of $v$
that is not in the $x$-range of the channels. 
Let $s$ be a horizontal sliding 2-transmitter that guards $p$.  Then
$s$ must be in one of the three boxes of $v$.

Assume now that exactly one sliding 2-transmitter intersects the vertex-gadget of $v$,
and it is not in the middle box.  Say the sliding 2-transmitter is in the bottom box.  
If $v$ has any neighbour $w$ above, then let $p$ be a point in the top box
of $v$ and in the same
$x$-range as the edge-gadget of $(v,w)$.  To guard $p$, we need either a
sliding 2-transmitter in the edge-gadget (which was excluded) or in the top or
middle box of $v$ (which was also excluded).  So $v$ cannot have any neighbour
above.  By construction that means that $v$ is the topmost of all vertices.
To guard the top box of $v$, we then must have a sliding 2-transmitter in the top or
middle box of $v$; again contradiction.
\end{proof}

\begin{lemma}
\label{lem:nphard_disconnected}
The following statements are equivalent: \begin{inparaenum}[(i)] \item $G$ has a vertex cover of size $k$, \item $P'$ can be guarded by $n+k$
sliding 2-transmitters, and \item $P'$ can be guarded by $n+k$ horizontal sliding 2-transmitters. \end{inparaenum}
\end{lemma}
\begin{proof}
Given a vertex cover $C$ of $G$, we place horizontal transmitters as follows:
If $v\in C$, then place a maximal horizontal sliding 2-transmitter in both outer boxes of $v$,
else place a maximal
horizontal sliding 2-transmitter in the middle box of $v$.  Clearly we have $n+|C|$
sliding 2-transmitters and every vertex-gadget is guarded.  For every edge $e$,
one endpoint $v$ is in $C$, and hence both bottom and top box of $v$
contain sliding 2-transmitters.  The one in the outer box of $v$ that
faces $e$ then guards the edge-gadget of $e$.

Vice versa, assume that set $S$ of sliding 2-transmitters guards $P'$.  By the above
results, we may assume that they are all horizontal and none are in an
edge-gadget.  Define $C$ to be all those vertices whose
vertex-gadgets are intersected by at least two sliding 2-transmitters.  Since every vertex-gadget
intersects at least one sliding 2-transmitter 
we have $|C|\leq |S|-n$.  
For every edge $(v,w)$,
the edge-gadget must be guarded by a sliding 2-transmitter that is in an outer box
of $v$ or $w$, say $v$.  Then $v$ must contain at least two sliding 2-transmitters by Lemma \ref{lem:vertex_transmitters},
so $v\in C$.  Hence $C$ is a vertex cover.
\end{proof}

\paragraph{Connecting the Polygon.}
Now we explain how to make the polygon connected while staying monotone.
Let $g_1,\dots,g_{m+n}$ be the gadgets in $P'$, sorted in bottom-to-top
order (since $y$-ranges
are disjoint, this is well-defined). The idea is to connect each $g_i$
to $g_{i+1}$ 
using a {\em connector-gadget}.  This is an $S$-shaped or
$Z$-shaped gadget much like a vertex-gadget, except that the top and
bottom box both add a zig-zag near the end.   
Also, one of the channels
has flexible height, so that the connector-gadget can have arbitrary
height.  We attach the ends of the connector-gadget $C$ to corners of
$g_i$ and $g_{i+1}$.
Fig.~\ref{fig:connectorGadget} shows how to do this if 
the $x$-range of $C$ is disjoint (except at the ends) from the ones
of $g_i$ and $g_{i+1}$, and the inset in 
Fig.~\ref{fig:subdividingAnEdge} shows how to do this if $C$
shares $x$-range with them (in case of which we push the
zig-zag to the very end to avoid overlap.)

\begin{figure}[t]
\centering
\includegraphics[width=0.70\textwidth, page=1]{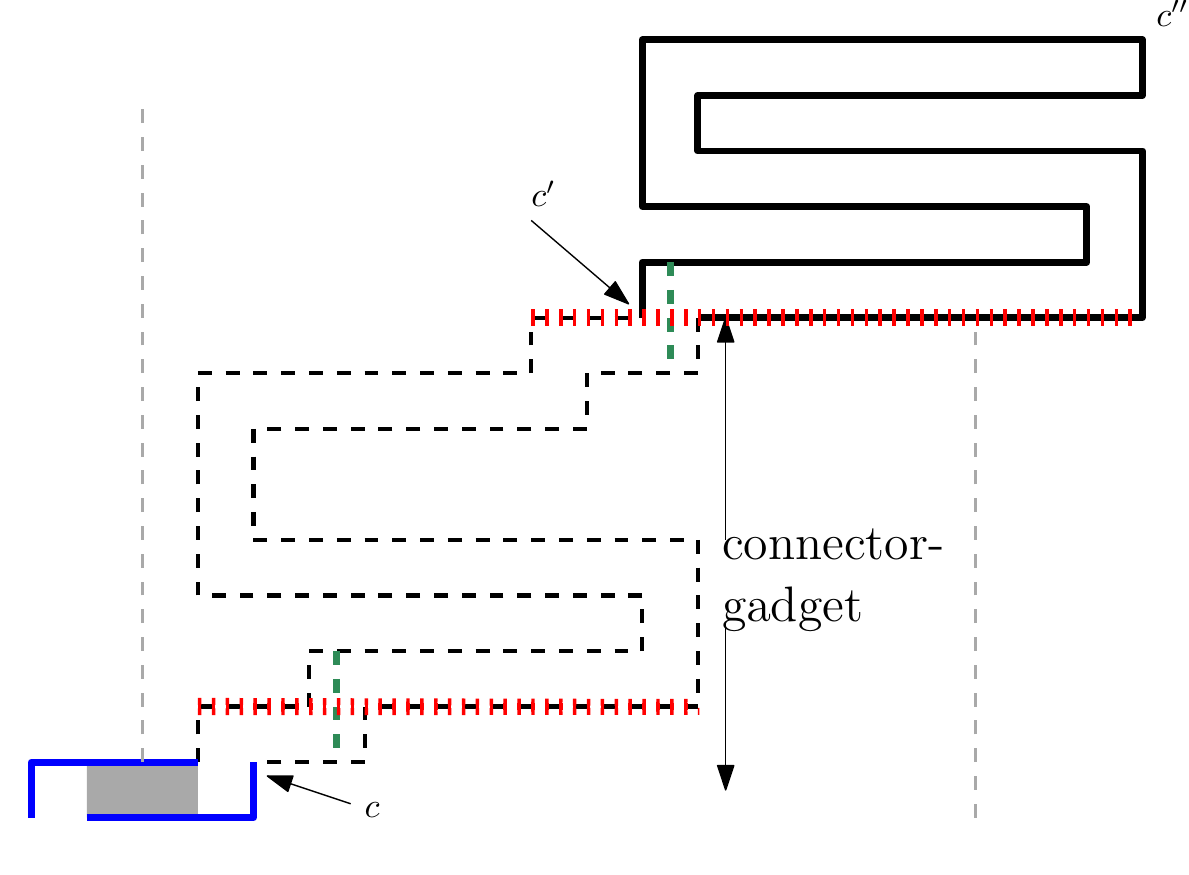}
\caption{Connecting an edge-gadget to a vertex-gadget if there is
no line of sight between them.  We again show how some vertical transmitters
can be replaced by horizontal transmitters.}
\label{fig:connectorGadget}%
\end{figure}

However, we cannot connect consecutive gadgets if the connector-gadget would
cross a line-of-sight.  To avoid doing this, we will subdivide edges.

\begin{observation}[Folklore] If $G^s$ results from graph $G$ by subdividing
one edge twice, then $G$ has a vertex cover of size $k$ if and only if $G^s$
has a vertex cover of size $k+1$.
\end{observation}

We proceed as follows. First ``parse'' the 
bottommost gadget $g_1$: use an $S$-shape for it
and fix as current corner its top right corner.  Assume now we have
parsed gadget $g_i$ already, and fixed one top corner $c$ of it as
current corner.  Let $g_{i+1}$ be the next gadget above $g_i$.  Considering
its two bottom corners, we choose the corner $c'$ so that
$\overline{cc'}$ crosses as few lines-of-sight as possible.

If line segment $\overline{cc'}$ crosses no line-of-sight, then 
attach a connector-gadget between $c$ and $c'$, 
using as shape (i.e., $S$ or $Z$) the one that has $c$ and $c'$ at its ends.
Let $c''$ be
the diagonally opposite corner from $c'$ in gadget $g_{i+1}$ and
(if $g_{i+1}$ is a vertex-gadget) 
use as shape (i.e., $S$ or $Z$) for it the one that has $c'$ and $c''$ at its ends.
This finishes (in this case) the parsing of gadget $g_{i+1}$,
and we continue to connect to the next gadget with current corner $c''$.

Now assume that $\overline{cc'}$ crosses some lines-of-sight, say
$l_1,\dots,l_\ell$ in order from $c$ to $c'$.  
For all $j$, line-of-sight 
$l_j$ represents an edge $e_j$; subdivide $e_j$
twice.  This adds two new vertex-gadgets and two new edge-gadgets
that we place along $l_j$, in the $y$-range between $g_i$ and
$g_{i+1}$.  We make their height small enough and move 
them up and down suitably (while staying
between $g_i$ and $g_{i+1}$), so that all their $y$-ranges are
disjoint and the ones of $l_j$ are below the ones of $l_{j+1}$ for all $j$.

All these gadgets can be connected with line segments that do
not cross a line-of-sight.  See 
Fig.~\ref{fig:subdividingAnEdge}.  We can hence connect all these gadgets 
as explained above.  The only difference is that the next current corner 
$c''$ must be chosen to be the end of the line segment connecting to
the next gadget.  Normally $c''$ will
again be diagonally opposite from the previous corner $c'$, 
but there is one exception per
set of gadgets added for subdivisions. With that we have connected to $g_{i+1}$,
and we repeat from there (after choosing its shape and
the current corner as before).

\begin{figure}[t]
\centering
\includegraphics[width=0.75\textwidth, page=2]{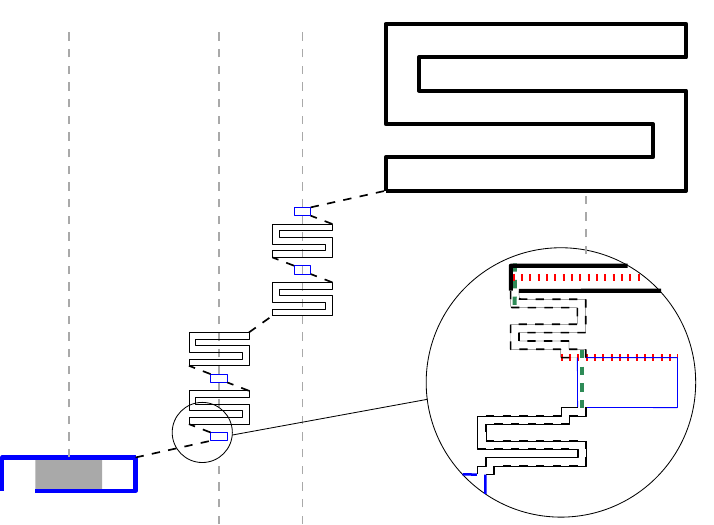}
\caption{Connecting an edge-gadget to a vertex-gadget if there are
lines of sight between them.}
\label{fig:subdividingAnEdge}%
\end{figure}

\paragraph{Reduction Revisited.}
With the addition of connector-gadgets, Observation~\ref{obs:no_vertical}
(vertical transmitters can be replaced by horizontals) is not as obvious 
anymore, but still holds as long as sliding $k$-transmitters may run along
polygon-edges. See Fig.~\ref{fig:connectorGadget}.  
With this, Lemma~\ref{lem:transmitters_noEdge}, 
Lemma~\ref{lem:vertex_transmitters}, and the equivalent of
Lemma~\ref{lem:vertex_transmitters} for connector-gadgets, 
also hold.
Let $N_s$ be the total number of subdivisions that we did over all
connecting of all gadgets ($N_s$ is even), and let $G'$ be the graph that 
results.  We started with $n+m\in O(n)$ vertex-gadgets and edge-gadgets
and $2m\in O(n)$ lines-of-sight.  Connecting two of these gadgets hence
creates $O(n)$ subdivisions, and therefore $N_s\in O(n^2)$ is polynomial.
After all subdivisions we have $n+m+2N_s$ gadgets, and hence use
$N_c:=n+m+2N_s-1$ connector-gadgets to connect all of them into one
polygon $P$.  So the construction is polynomial.  $G$ has a vertex cover
of size $k$ if and only if $G'$ has a vertex cover of size $k':=k+N_s/2$
if and only if $P$ can be guarded with $k'+n+N_c$ horizontal sliding
$2$-transmitters.

With that, the reduction is complete for $k=2$.
Note that the constructed polygon is connected and $y$-monotone (and in 
particular therefore simple).

\subsection{$k$-Transmitters for $k>2$}
We now generalize to sliding $k$-transmitters for any fixed $k>0$.
The reduction is exactly the same as before, with the exception of
the definition of vertex-gadgets and connector-gadgets.

The vertex-gadget now consists of 
$k+1$ copies of the thickened bar in the visibility representation
(earlier we had $3=2+1$ copies).  They are connected
with $k$ channels at alternate ends, resulting in a zig-zag line.  The
connector-gadget is a vertex-gadget with additional small zig-zags in
the top and bottom box (possibly pushed towards the end.)
See Fig.~\ref{fig:gadgetsForK}.

\begin{figure}[t]
\centering
\includegraphics[width=0.6\linewidth]{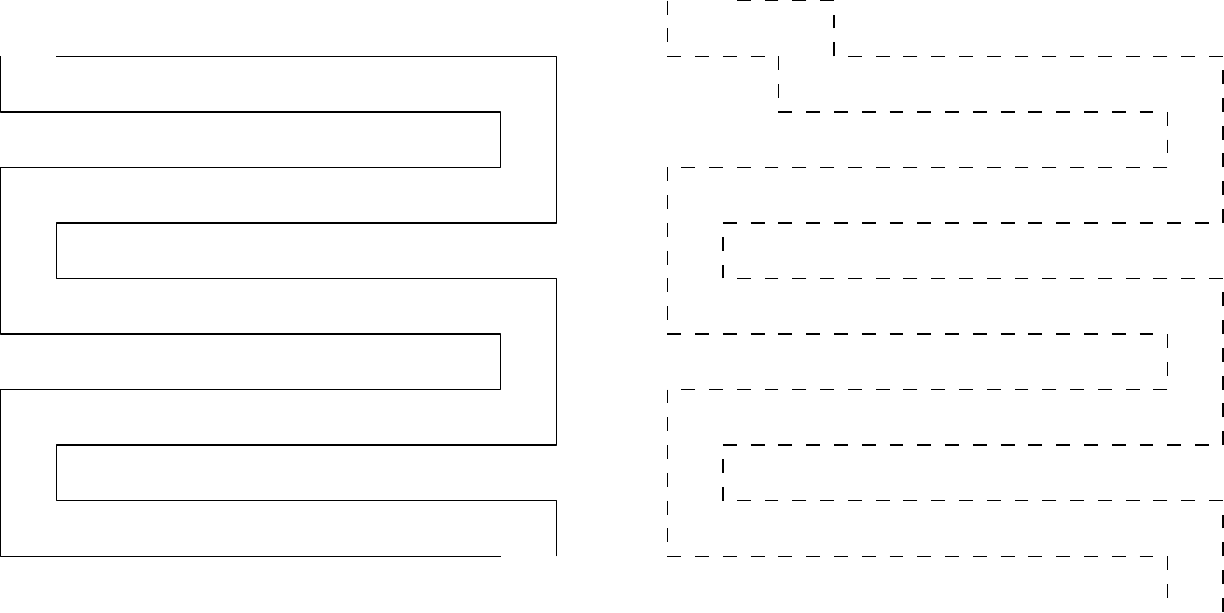}
\caption{Vertex- and connector-gadget for $k=4$.}
\label{fig:gadgetsForK}%
\end{figure}

We can verify that again vertical sliding $k$-transmitters are never
better than horizontal ones.  Define for a vertex-gadget the {\em middle}
box to be the $(k/2{+}1)$st box (recall that $k$ is even), and the
{\em outer boxes} to be the top and bottom box as before.  With that,
the proofs of
Lemmas~\ref{lem:transmitters_noEdge} and \ref{lem:vertex_transmitters} 
carry almost verbatim, and the reduction holds again.  We conclude:

\begin{theorem}
\label{thm:mainForKs}
For any $k>0$,
guarding a polygon with the minimum set of sliding $k$-transmitters
is \textsc{NP}-complete, even if (i) the polygon is a simple $y$-monotone 
orthogonal polygon, and (ii) only horizontal sliding $k$-transmitters are allowed.
\end{theorem}

Notice that every gadget is a {\em thickened path} obtained by sliding
a unit square along an orthogonal path.   With suitable rescaling,
in fact the entire polygon can be made into a thickened path, 
with one exception:  Whenever we subdivide edges, we must (at one edge-gadget)
attach both connecting gadgets on the same (left or right) side, hence have a
``leg'' sticking out.  (This could perhaps be called a {\em thickened
caterpillar}.)  We suspect that the construction could be modified to become a thickened path, but have not been able to work out
the details yet.

\section{An $O(1)$-Approximation Algorithm}
\label{sec:approximation}
In this section, we give an $O(1)$-approximation algorithm for the sliding $k$-transmitter problem, using as key ingredient an $O(1)$-approximation developed in \cite{BiedlCLMMV} for a certain hitting problem among segments. The main difference between our approach and the one in \cite{BiedlCLMMV} is that we need to define the segments differently so that we encapture that guards can see through $k$ walls.

Let $P$ denote an orthogonal polygon with $n$ vertices, and let $\delta P$ denote the boundary of $P$. We first compute a subdivision of $P$ and then define sets of orthogonal line segments in $P$ which will be used to define the hitting set problem.

\paragraph{Slices.}
Define \emph{horizontal partition-segments} as
follows: Start with a horizontal edge $e$.  Expand $e$ leftwards until
we hit a vertical edge of the polygon, coming from the strict inside
of $P$, for the $(k/2)$th time. Likewise expand $e$ rightwards.  
If there are not enough such intersections, then stop at the last one.
See Fig.~\ref{fig:segments}.

\begin{figure}[t]
\centering
\hspace*{\fill}
\includegraphics[width=.30\textwidth,page=1]{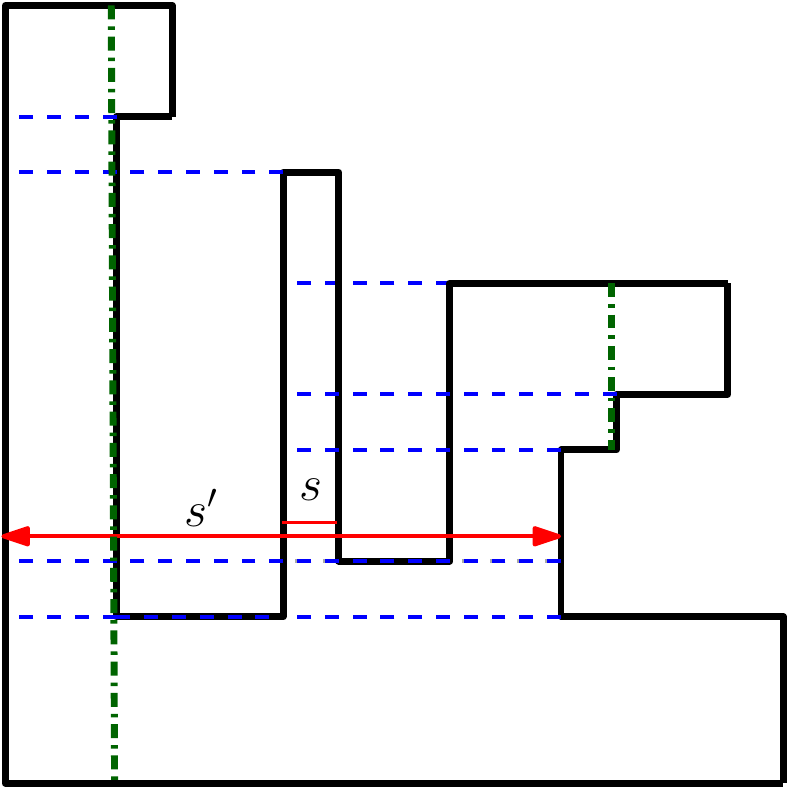}
\hspace*{\fill}
\caption{Horizontal partition-segments (blue dashed), 
one horizontal slice-segment (red solid) and two vertical
guard-segments (green dot-dashed) for $k=2$.}
\label{fig:segments}
\end{figure}

The horizontal partition-segments split the interior of the polygon into
rectangles that we call {\em horizontal slices}.
Since any edge gives rise to one partition-segment, and any partition-segment
intersects $O(k)$ vertical edges, 
we have $O(kn)$ horizontal slices.

The following lemma argues that this partitioning is ``correct'' in the
sense that any transmitter either guards all or nothing of the interior of
a slice.

\begin{lemma}
\label{lem:seeingBEntirely}
Let $\sigma$ be a horizontal slice and let $c$ be a point in its interior.  
If a maximal vertical sliding $k$-transmitter $g$ sees $c$, 
then it sees all points of $\sigma$.
\end{lemma}
\begin{proof}
Let $q$ be the point on $g$ where the perpendicular 
from $c$ onto $g$ ends. 
By assumption the horizontal line segment $\overline{cq}$ intersects the 
boundary of $P$ at most $k$ 
times.  Expand $\overline{cq}$ until it spans the $x$-range of $\sigma$; 
this cannot add crossings since $\sigma \subseteq P$.  
Now sweep the resulting segment $s'$ upward until we hit either the top
side of $\sigma$ or a horizontal edge of $P$.  Say we hit an edge $e$ first.
We created a partition-segment from $e$, which extends in both directions
until it hits at most $(k/2)$ vertical edges from the inside, hence at
most $k$ vertical edges.  This partition segment contains the entire
(translated) $s'$ and splits $B$, so we have reached the top side of $B$.

So we can sweep the region of $B$ above $s'$ without encountering
new edges of $P$, which shows that $g$ guards all of this.  Likewise
we can sweep downward until the bottom side of $B$, and so $g$ guards all
of $B$.
\end{proof}

\paragraph{Slice-segments. } We now assign a segment to each
horizontal slice that captures ``being guarded''.  For any horizontal
slice $\sigma$, let $s$ be a horizontal segment strictly inside $\sigma$.
Extend $s$ (much like we did for partition-segments) to both sides until
it hits a vertical edge from the inside for the $(k/2)$th time.
We call the resulting segment $s'$ the \emph{slice-segment} of $\sigma$.
See also Fig.~\ref{fig:segments}.  

We define vertical slices of $P$
and vertical slice-segments in an analogous fashion.  
There are $O(kn)$ slice-segments
since there are $O(kn)$ slices.

\paragraph{Guard-segments $\Gamma$.} Our definition of sliding
$k$-transmitters allowed any horizontal or vertical segment to be used as
such.  We now describe a finite set of sliding $k$-transmitters  and argue
that these suffice.  Let $s$ be a horizontal edge of $P$.  Define
a sliding $k$-transmitter $s'$ obtained by extending $s$ until we hit an
interior point of a vertical edge of $\delta P$.  (If some vertices are
aligned, then $s'$ may run along multiple horizontal edges of $P$.)
The resulting segments are
the {\em horizontal guard-segments} $\Gamma_H$.  Define {\em vertical
guard-segments} $\Gamma_V$ similarly, and set $\Gamma=\Gamma_H\cup \Gamma_V$
to be the guard-segments.  
We have at most $n$ guard segments (one per edge).

\paragraph{Crosses $X$.}  Let a {\em pixel} be any
rectangle that has the form $\sigma_H\cap \sigma_V$ for a horizontal
slice $\sigma_H$ and vertical slice $\sigma_V$.  Let $c$ be the point
where the slice-segments $s_H,s_V$ corresponding to $\sigma_H$ and $\sigma_V$
intersect; we call $c$ a {\em cross}, and say that $s_H$ and $s_V$ 
\emph{support} $c$.   Note that $c$ is in the interior of the pixel
since slice-segments were defined using segments strictly in the interior
of the slice.
We denote the set of crosses by $X$.  

We say that a cross $c$ is {\em hit} by a guard-segment $g$ if $g$ intersects 
the supporting slice-segment of $c$ that is perpendicular to $g$.
We now show that reducing the problem to just crosses and guard-segments is
enough.

\begin{lemma}
\label{lem:problemReduction}
A set $S$ of $m$ sliding $k$-transmitters guards $P$ if and only if there exists a set $S'\subseteq \Gamma$ of $m$ guard-segments such that every cross $c$ is hit by some guard-segment $\gamma\in S'$.
\end{lemma}
\begin{proof}
$(\Rightarrow)$ Suppose that we have a set $S$ of $m$ sliding $k$-transmitters that guards $P$ entirely.   Fix one sliding $k$-transmitter $s$.
Translate $s$ in parallel (i.e., move it horizontally if $s$ is vertical, move it vertically if $s$ is horizontal) until we reach $\delta P$.
Thus $s$ is now intersecting an edge of $P$.  Extend
$s$ so that it is maximal while still within $P$.  
Both operations can only increase the region seen.  The resulting segment $s'$ is a guard-segment.  After doing this to all sliding $k$-transmitters, we now have a set of guard-segments $S'$
that sees all of $P$.
Now consider any cross $c\in X$. Since $c$ is a point in $P$, it is guarded by some guard-segment $\gamma\in S'$. Thus, there exists a point $g\in \gamma$ such that the line segment $gc$ is normal to $\gamma$ and intersects $\delta P$ in at most $k$ points. But, $gc$ is part of the slice-segment that supports $c$ and is perpendicular to $\gamma$.  So $g$ is the intersection point between that slice-segment and guard-segment $\gamma$.

$(\Leftarrow)$ 
This is straightforward by Lemma~\ref{lem:seeingBEntirely} since
\begin{inparaenum}[(i)] 
\item any point in $P$ belong to at least one pixel,
\item there is a 1-to-1 correspondence between the pixels and crosses of $P$, 
and \item crosses are interior points of pixels. 
\end{inparaenum}
\end{proof}

Our problem has now been discretized as follows:  Given the set $X$ of
all crosses, each supported by two line segments, find a 
subset $S\subseteq \Gamma$ such that for every cross one of the two
line segments is intersected by at least one guard-segment in $S$.  
We call this the {\em cross-hitting} problem.
This problem is {\em exactly} the same problem as
studied in \cite{BiedlCLMMV} when solving the sliding-cameras problem
(the only difference is in the choice of 
supporting segments of crosses, which are longer for sliding $k$-transmitters).
They give an $O(1)$-approximation algorithm for this problem which uses
no information about how the segments were obtained (other than that they
are horizontal or vertical).  Using this, we hence have:

\begin{theorem}
\label{thm:mH2TApprox}
For any $k>0$, there exists a polynomial-time $O(1)$-approximation algorithm for
guarding an orthogonal polygon with sliding $k$-transmitters.
\end{theorem}

As in \cite{katz2011}, we also consider the variant when only horizontal sliding
$k$-transmitters are allowed.  This also reduces to the cross-hitting
problem, with the only change that we use
$\Gamma_H$ in place of $\Gamma$. 
This in fact simplifies the problem, because now only
vertical supporting segments are relevant for crosses.  So there is
also an $O(1)$-approximation algorithm for guarding an orthogonal
polygon with horizontal sliding $k$-transmitters.

\section{Conclusion}
\label{sec:conclusion}
In this paper, we studied how to guard an orthogonal polygon using the minimum number of
sliding $k$-transmitters.  We showed that this is NP-hard, even if
the polygon is $y$-monotone.   We also gave an $O(1)$-approximation
algorithm.

The main open problem is to find better approximation factors.  
(The ``$O(1)$'' in \cite{BiedlCLMMV} stems from the use of
$\varepsilon$-nets, and the constant is unspecified but likely quite 
large.)  Is the problem APX-hard? Also, for what subclass of polygons is guarding with sliding $k$-transmitters
polynomial?  This is true for orthogonally convex polygons (one or two
guards are always enough), but are there other, less trivial classes?

\bibliographystyle{plain}
\bibliography{ref}

\end{document}